\documentclass[12pt,oneside]{article}
\usepackage[osf]{libertine}

\usepackage[OT1]{fontenc}
\usepackage[utf8]{inputenc}
\usepackage{kpfonts}
\usepackage{graphicx}
\usepackage[dvipsnames]{xcolor}
\usepackage{amsthm}
\usepackage[margin=1.1in]{geometry}
\usepackage[sf,bf,small,raggedright,compact]{titlesec}
\usepackage{hyperref}
\usepackage{mathrsfs}
\usepackage{amssymb}
\usepackage{stmaryrd}
\usepackage[most]{tcolorbox}

\definecolor{linkblue}{named}{MidnightBlue}
\hypersetup{colorlinks=true, linkcolor=linkblue,  anchorcolor=linkblue,
        citecolor=linkblue, filecolor=linkblue, menucolor=linkblue,
        urlcolor=linkblue}
\usepackage[capitalize]{cleveref}
\usepackage{paralist}
\usepackage[longnamesfirst,numbers,sort&compress]{natbib}

\usepackage{listings,newtxtt}
\newtheorem{thm}{Theorem}

\newtheorem{clm}{Claim}

\newtheorem{remark}{Remark}

\usepackage{dsfont}  
\usepackage[noend]{algorithmic}

\usepackage[normalem]{ulem}
\usepackage{cancel}
\usepackage{enumitem}

\usepackage{todonotes}

\usepackage[longnamesfirst,numbers,sort&compress]{natbib}

\usepackage[mathlines]{lineno}
\newcommand*\patchAmsMathEnvironmentForLineno[1]{%
 \expandafter\let\csname old#1\expandafter\endcsname\csname #1\endcsname
 \expandafter\let\csname oldend#1\expandafter\endcsname\csname end#1\endcsname
 \renewenvironment{#1}%
    {\linenomath\csname old#1\endcsname}%
    {\csname oldend#1\endcsname\endlinenomath}}%
\newcommand*\patchBothAmsMathEnvironmentsForLineno[1]{%
 \patchAmsMathEnvironmentForLineno{#1}%
 \patchAmsMathEnvironmentForLineno{#1*}}%
\AtBeginDocument{%
\patchBothAmsMathEnvironmentsForLineno{equation}%
\patchBothAmsMathEnvironmentsForLineno{align}%
\patchBothAmsMathEnvironmentsForLineno{flalign}%
\patchBothAmsMathEnvironmentsForLineno{alignat}%
\patchBothAmsMathEnvironmentsForLineno{gather}%
\patchBothAmsMathEnvironmentsForLineno{multline}%
}

\usepackage{skull}
\usepackage{paralist}
\usepackage{graphicx}
\usepackage[noend]{algorithmic}

\usepackage[normalem]{ulem}
\usepackage{cancel}
\usepackage{comment}

\makeatletter
\renewcommand{\fnum@figure}{Fig. \thefigure}
\makeatother

\definecolor{brew1}{rgb}{0.552941176471, 0.827450980392, 0.780392156863}
\definecolor{brew2}{rgb}{1.0, 1.0, 0.701960784314}
\definecolor{brew3}{rgb}{0.745098039216, 0.729411764706, 0.854901960784}
\definecolor{brew4}{rgb}{0.98431372549, 0.501960784314, 0.447058823529}
\definecolor{brew5}{rgb}{0.501960784314, 0.694117647059, 0.827450980392}
\definecolor{brew6}{rgb}{0.992156862745, 0.705882352941, 0.38431372549}
\definecolor{brew7}{rgb}{0.701960784314, 0.870588235294, 0.411764705882}
\definecolor{brew8}{rgb}{0.988235294118, 0.803921568627, 0.898039215686}
\definecolor{ultramarine}{rgb}{0, 0.125, 0.376}

\title{\fontsize{19}{24}\selectfont Deciding if a DAG is Interesting is Hard}

\author{
\normalsize{Jean-Lou De Carufel}\thanks{School of Electrical Engineering and Computer Science, University of Ottawa, Ottawa, Canada. Research supported by NSERC.} \and 
\normalsize{Anil Maheshwari}\thanks{School of Computer Science, Carleton University, Ottawa, Canada. Research supported by NSERC.} \and 
\normalsize{Saeed Odak}\footnotemark[1] \and 
\normalsize{Bodhayan Roy}\thanks{Department of Mathematics, Indian Institute of Technology Kharagpur, India. Research supported by SERB MATRICS Grant Number MTR/2021/000474.} \and
\normalsize{Michiel Smid}\footnotemark[2] \and 
\normalsize{Marc Vicuna}\footnotemark[2]
}

\date{}

\begin{document}

\maketitle

\begin{abstract}

The \emph{interestingness score} of a directed path $\Pi = e_1, e_2, e_3, \dots, e_\ell$ in an edge-weighted directed graph $G$ is defined as $\texttt{score}(\Pi) := \sum_{i=1}^\ell w(e_i) \cdot \log{(i+1)}$, where $w(e_i)$ is the weight of the edge $e_i$. We consider two optimization problems that arise in the analysis of Mapper graphs, which is a powerful tool in topological data analysis.
In the IP problem, the objective is to find a collection $\mathcal{P}$ of edge-disjoint paths in $G$ with the maximum total interestingness score. 
For $k \in \mathbb{N}$, the $k$-IP problem is a variant of the IP problem with the extra constraint that each path in $\mathcal{P}$ must have  exactly $k$ edges. Kalyanaraman, Kamruzzaman, and Krishnamoorthy (Journal of Computational Geometry, 2019) claim that both IP and $k$-IP (for $k \geq 3$) are NP-complete. We point out some inaccuracies in their proofs. 
Furthermore, we show that both problems are NP-hard in directed acyclic graphs. 
\end{abstract}

\section{Introduction}

The Mapper algorithm, introduced by 
Singh, M{\'{e}}moli, and Carlsson
\cite{DBLP:conf/spbg/SinghMC07}, is one of the powerful methods in topological data analysis that creates a simplified graph representation of high-dimensional data.\footnote{For a precise description of the Mapper algorithm and its applications in topological data analysis, see  \cite{coskunuzer2024topological, DBLP:conf/spbg/SinghMC07} and the references therein.} In \cite{DBLP:journals/jocg/KalyanaramanKK19}, 
Kalyanaraman, Kamruzzaman, and Krishnamoorthy 
introduce a collection of combinatorial optimization problems to quantify the interestingness of subsets of vertices in Mapper graphs. Formally, let $G = (V, E)$ be an edge-weighted directed graph with a weight function $w : E \rightarrow \mathbb{N}$. Moreover, each edge of $G$ is assigned an integer signature by a function $\texttt{sig} : E \rightarrow \mathbb{N}$. Let $\Pi = e_1, e_2, e_3, \cdots, e_k$ be a directed \emph{$k$-path} (i.e., a directed path having $k$ edges) in $G$, where $e_i$ is an edge ($1 \le i \le k$). 
The path $\Pi$ is \emph{interesting} if $\texttt{sig}(e_i)$ is identical for all $1\le i \le k$. The \emph{interestingness score} of an interesting path $\Pi$ is defined as 
\[\texttt{score}(\Pi) := \sum_{i=1}^{k} w(e_i)\cdot \log{(i+1)}.
\]
(All logarithms in this paper are in base 2.) The following three optimization problems, on a given edge-weighted directed graph $G$ with a signature function $\texttt{sig}$, have been introduced by 
Kalyanaraman, Kamruzzaman, and Krishnamoorthy \cite{DBLP:journals/jocg/KalyanaramanKK19}.

\begin{itemize}
    \item \emph{Max-IP Problem}: Find an interesting path $\Pi$ in $G$ such that $\texttt{score}(\Pi)$ is maximized.
    \item \emph{IP Problem}: Find a collection $\mathcal{P}$ of edge-disjoint interesting paths in $G$ such that the sum of the interestingness scores of the paths in $\mathcal{P}$ is maximized.
    \item \emph{$k$-IP Problem}: Find a collection $\mathcal{P}$ of edge-disjoint interesting $k$-paths in $G$ such that the sum of the interestingness scores of the paths in $\mathcal{P}$ is maximized.
\end{itemize}

Observe that one can consider the edge-disjoint graphs induced by edges with the same signature and solve these optimization problems separately. Therefore, from now on, we assume that the signatures are identical for all edges in $G$. We define a \emph{path partition} of a directed graph $G$ to be a collection $\mathcal{P}$ of edge-disjoint paths in $G$ such that each edge of $G$ appears in exactly one path of $\mathcal{P}$. Let $\mathcal{C}$ be an arbitrary collection of paths in $G$. The \emph{interestingness score} of $\mathcal{C}$, denoted by $\texttt{score}(\mathcal{C})$, is defined as the sum of the interestingness scores of the paths in $\mathcal{C}$. That is
\[
    \texttt{score}(\mathcal{C}) := \sum\limits_{\Pi \in \mathcal{C}} \texttt{score}(\Pi) = \sum_{\substack{\Pi \in  \mathcal{C}  \\ \Pi = e_1,e_2,\dots,e_k}} \sum_{i=1}^{k} w(e_i)\cdot \log{(i+1)}. 
\]

Given this definition of the interestingness score, we consider the following equivalent formulation of the aforementioned problems.

\begin{itemize}
    \item \emph{Max-IP Problem}: Find a path $\Pi$ in $G$ such that $\texttt{score}(\Pi)$ is maximized.
    \item \emph{IP Problem}: Find a path partition $\mathcal{P}$ of $G$ such that $\texttt{score}(\mathcal{P})$ is maximized.
    \item \emph{$k$-IP Problem}: Find an edge-disjoint set $\mathcal{P}$ of $k$-paths in $G$ such that $\texttt{score}(\mathcal{P})$ is maximized.
\end{itemize}


Kalyanaraman et al. \cite{DBLP:journals/jocg/KalyanaramanKK19} claim the following results: First, the Max-IP problem is NP-complete, even when all edge weights in the graph are equal to $1$. Second, the Max-IP problem can be solved in polynomial time when the graph $G$ is a directed acyclic graph (DAG). Third, on a DAG, the $k$-IP problem is solvable in polynomial time for $k = 2$, whereas for $k \geq 3$, the problem is NP-complete. 

The authors in \cite{DBLP:journals/jocg/KalyanaramanKK19} assume that the edge weights can be arbitrary real numbers. Moreover, they assume that the logarithm function can be computed in $O(1)$ time. Recall that in complexity theory, the model of computation used is a Turing machine, or any model belonging to the First Machine Class; see van Emde Boas \cite{van1991handbook} for an excellent overview. Moreover, the running time of algorithms is measured by the number of \emph{bit-operations}. Thus, input values are assumed to be integers, or rational numbers, and the running time (i.e., the number of bit-operations) is bounded by a function of the total number of bits in the binary representations of all input values. 
Because of this, it is not clear that the three problems defined above are in the complexity class NP. 

We also point out that there is a gap in their reduction for the NP-hardness proof of the $k$-IP problem; see \cite[Section 4.4.3]{marcsthesis}. 

The authors in \cite{DBLP:journals/jocg/KalyanaramanKK19} also conjecture that the IP problem is NP-complete. For the same reasons as above, it is not clear if this problem is in NP. 

In this paper, we prove that both the IP problem and the $k$-IP problem (for any fixed $k \geq 3$) are NP-hard when the graph $G$ is an edge-weighted directed acyclic graph, and all edge weights are natural numbers.

Throughout the rest of this paper, \emph{time}s always refers to the number of bit-operations.

\section{Hardness of the IP Problem} \label{sec:ip}

In this section, we show that the IP problem on directed acyclic graphs with positive integer weights is NP-hard. Consider the following decision version of the IP problem.

\begin{tcolorbox}[colback=white, colframe=teal, sharp corners=all, boxrule=0.25mm, rounded corners=all, width=\textwidth]
\textbf{IP Problem} \textbf{(decision version)}\\
\textbf{Input:} An edge-weighted directed acyclic graph $G=(V,E)$, a weight function $w : E \rightarrow \mathbb{N}$, and a positive integer $t$. \\
\textbf{Question:} Does there exist a path partition $\mathcal{P}$ of $G$ such that the interestingness score of $\mathcal{P}$ is at least $\log{t}$?
\end{tcolorbox}

\begin{remark}
The authors in \cite{DBLP:journals/jocg/KalyanaramanKK19} use $t$ instead of $\log t$. As we will see, by using $\log t$, our reduction takes polynomial time. 
\end{remark}

To prove NP-hardness, we reduce, in polynomial time, the $3$-SAT problem to the IP Problem. 

\begin{tcolorbox}[colback=white, colframe=teal, sharp corners=all, boxrule=0.25mm, rounded corners=all, width=\textwidth]
\textbf{$3$-SAT Problem}\\
\textbf{Input}: A set $X = \{x_1,x_2,...,x_n\}$ of Boolean variables and a collection $\mathcal{C} = \{C_1,C_2,...,C_m\}$ of clauses
over $X$ such that each clause is a disjunction of at most three literals. \\
\textbf{Question}: Does there exist a truth assignment to the $n$ variables in $X$ such that the Boolean expression $C_1 \land C_2 \land \cdots \land C_m$ evaluates to true?
\end{tcolorbox}

\begin{thm} \label{thm:ip}
    The IP problem is NP-hard.
\end{thm}

\begin{proof}
    Let the Boolean formula $f = C_1 \land C_2 \land C_3 \land \cdots \land C_m$ over the variables $x_1, x_2, x_3, \cdots, x_n$ be an instance of the $3$-SAT problem. We assume that no clause of $f$ contains both the positive and negative literal of the same variable. Moreover, we assume each clause contains exactly three literals (if a clause has fewer than three literals, we repeat a literal). For each $1 \le i \le n$, let $\ell_i$ be the number of different clauses that contain the variable $x_i$. Then, $\sum_{i=1}^n \ell_i = 3m$.
    
    Initially, we introduce an unweighted directed acyclic graph $G_f$ corresponding to the formula $f$. Then using $G_f$, we create an edge-weighted directed acyclic graph $G^*_f$ as an instance of the IP problem. Consider a variable $x_i$ ($1 \le i \le n$). We define a variable gadget $G_i \subset G_f$ corresponding to $x_i$ in the following way. Let $\alpha_1, \alpha_2, \dots, \alpha_{\ell_i}$ be the indices of the clauses in which $x_i$ or $\overline{x}_i$ appear. Moreover, let $\alpha_{\ell_{i}+1} = \alpha_{\ell_1}$. The variable gadget $G_i$ corresponding to $x_i$ consists of $4\ell_i + 2$ vertices labeled:
    
    \begin{itemize}
        \item $s_i$ and $t_i$ (source and terminal vertices),
        \item $v_{i,\alpha_1}, v_{i,\alpha_2}, \dots, v_{i,\alpha_{\ell_i}}$,
        \item $u_{i,\alpha_1}, u_{i,\alpha_2}, \dots, u_{i,\alpha_{\ell_i}}$,
        \item $y_{i,\alpha_1}, y_{i,\alpha_2}, \dots, y_{i,\alpha_{\ell_i}}$ ($y$-vertices), and
        \item $\overline{y}_{i,\alpha_1}, \overline{y}_{i,\alpha_2}, \cdots, \overline{y}_{i,\alpha_{\ell_i}}$ ($\overline{y}$-vertices).
    \end{itemize}
    
    The edge set of $G_i$ is defined as follows. For each $1 \le j \le \ell_i$,
    
    \begin{itemize}
        \item $(s_i, v_{i,\alpha_j}) \in E(G_i)$,
        \item $(v_{i,\alpha_j}, y_{i,\alpha_j}) \in E(G_i)$,
        \item $(v_{i,\alpha_j}, \overline{y}_{i,\alpha_j}) \in E(G_i)$,
        \item $( {y}_{i,\alpha_j},u_{i,\alpha_j}), (\overline{y}_{i,\alpha_j},u_{i,\alpha_{j+1}}) \in E(G_i)$ (type-$\mathrm{U}$ edges), and
        \item $(u_{i,\alpha_j}, t_i) \in E(G_i)$ (type-$\mathrm{T}$ edges).   
    \end{itemize}

    Refer to \cref{fig:gadgets} for an example of a variable gadget. To complete the construction of $G_f$, for each $1 \le j \le m$, we consider a clause gadget $H_j \subset G_f$ corresponding to the clause $C_j$. Let $x_{\beta_1}, x_{\beta_2}, \text{ and } x_{\beta_3}$ be the three variables appearing in $C_j$. The subgraph $H_j$ consists of $3$ vertices, $z_{j,0}$, $z_{j,1}$, and $z_{j,2}$. The edge set of the clause gadget $H_j$ consists of the directed edges $(z_{j,0},z_{j,1})$ and $(z_{j,1},z_{j,2})$. Moreover, for each $1 \le p \le 3$, if the variable $x_{\beta_p}$ is positive in $C_j$, then we add the directed edge $(z_{j,2}, y_{\beta_p,j})$; otherwise, we add the directed edge $(z_{j,2}, \overline{y}_{\beta_p,j})$. This concludes the construction of the graph $G_f$.

    Finally, to construct the edge-weighted directed acyclic graph $G^*_f$, for each $1 \le j \le \ell_i$, we subdivide the edge $(s_i, v_{i,\alpha_j})$ twice, i.e., we replace the edge $(s_i, v_{i,\alpha_j})$ by the directed path $s_i = t_{i,\alpha_j,0}, t_{i,\alpha_j,1}, t_{i,\alpha_j,2}, t_{i,\alpha_j,3} = v_{i,\alpha_j}$. Observe that because of these subdivisions, all paths from $s_i$ to $t_i$ have $6$ edges. Each type-$\mathrm{T}$ edge has a weight equal to $29m$, and all other edges have a weight equal to $1$. Observe that, for each $1 \le i \le n$, $G^*_f$ contains exactly $\ell_i$ edge-disjoint paths from $s_i$ to $t_i$ each of which has interestingness score 
    \[ W = \log{(6!)} +  29m \cdot \log{7}. 
    \]
    This value is the largest possible interestingness score of a path in $G^*_f$. Moreover, there are exactly two ways to find these $\ell_i$ edge-disjoint paths, namely, $\ell_i$ edge-disjoint paths going through $y$-vertices and $\ell_i$ edge-disjoint paths going through $\overline{y}$-vertices.

    By construction, $G^*_f$ is acyclic. The instance for the IP problem is the graph $G^*_f$ with parameter 
    \[ t := 2^{m \cdot (\log{(5!)} + 5 \log{2} + 2 \log{3} + 3W)} = (2^5 \cdot 3^2 \cdot 5! \cdot (6!)^3)^m \cdot (7^{87})^{m^2} .
    \] 

\begin{figure}
    \centering
    \includegraphics[width=0.95\linewidth]{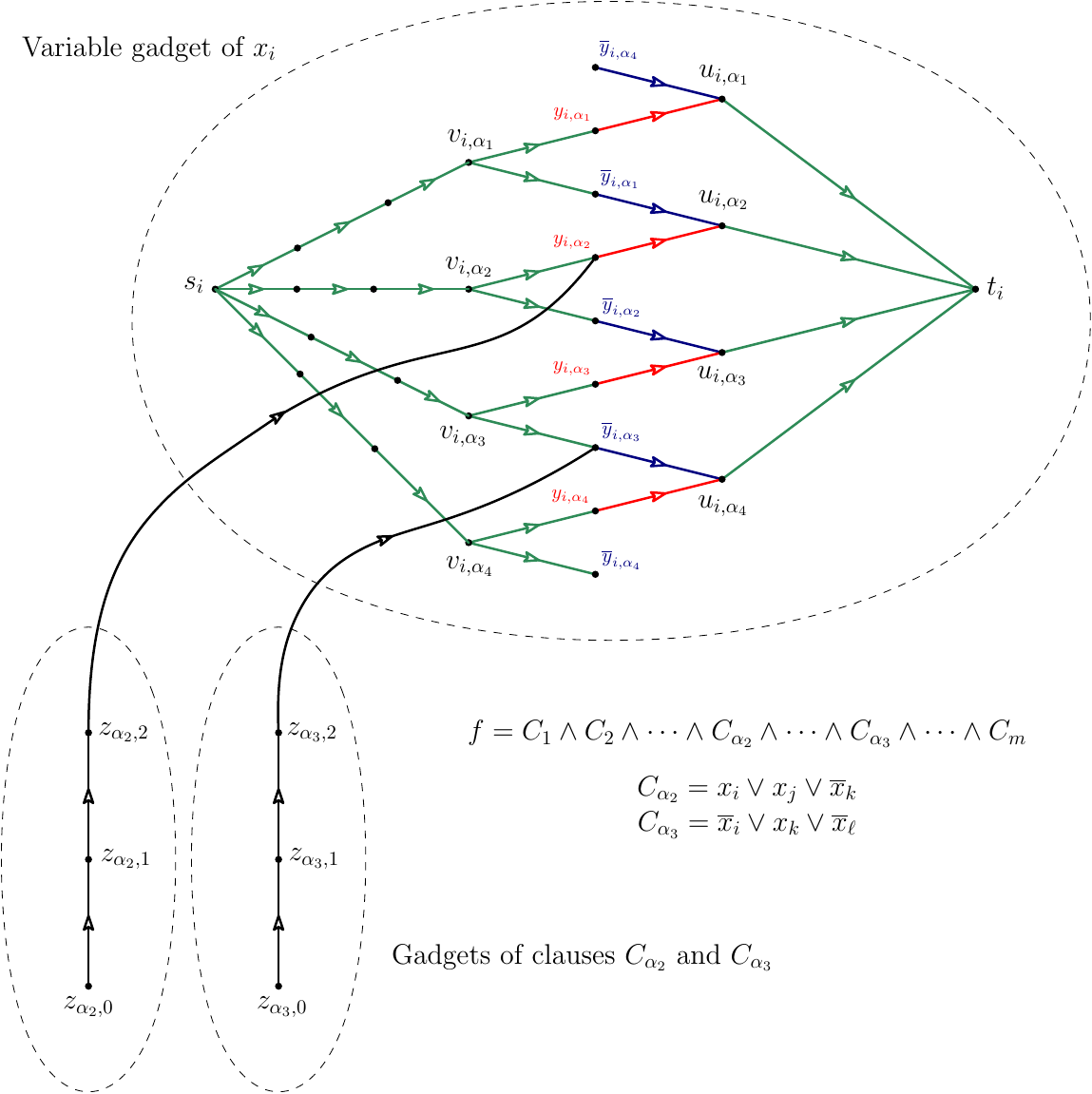}
    \caption{Let $f=C_1\land C_2\land\dots\land C_{\alpha_2} \land \dots \land C_{\alpha_3} \land \dots \land C_m$, where $C_{\alpha_2} = x_i \lor x_j \lor \overline{x}_k$ and $C_{\alpha_3} = \overline{x}_i \lor x_k \lor \overline{x}_{\ell}$. In this figure, we show the variable gadget for $x_i$ and the clause gadgets for $C_{\alpha_2}$ and $C_{\alpha_3}$. The vertices {\color{ultramarine} $\overline{y}_{i,\alpha_4}$} on the top and the bottom of the variable gadget are the same vertex (they are identified). The red and blue edges in the figure are type-$\mathrm{U}$ edges. They indicate the positive and negative appearances of the variable $x_i$, respectively.}
    \label{fig:gadgets}
\end{figure}

    \begin{clm}
    The formula $f$ is satisfiable if and only if $G^*_f$ contains an edge-disjoint path partition with an interestingness score of at least $\log{t}$.
    \end{clm}

    \begin{proof}
        
        Let $f$ be satisfiable. We construct a path partition $\mathcal{P}'$ of $G^*_f$ such that ${\texttt{score}(\mathcal{P}')} \ge \log t$. For each variable $x_i$, if $x_i$ is $\mathrm{TRUE}$, then we add the $\ell_i$ paths with interestingness score $W$ from $s_i$ to $t_i$ going through $\overline{y}$-vertices in $G_i$ to $\mathcal{P}'$. Otherwise, we add the $\ell_i$ paths with interestingness score $W$ from $s_i$ to $t_i$ going through $y$-vertices in $G_i$ to $\mathcal{P}'$.

        Moreover, for each clause, there is at least one literal with a $\mathrm{TRUE}$ assignment. Let $x_{i}$ be a $\mathrm{TRUE}$ literal in $C_{\alpha_j}$. If $x_{i}$ appears as a positive literal in $C_{\alpha_j}$, then none of the paths from $s_i$ to $t_i$ defined in the previous paragraph contains the edge $({y}_{i, \alpha_{j}}, u_{i,\alpha_{j}})$. In this case, we add the directed path from $z_{\alpha_j, 0}$ to $u_{i,\alpha_{j}}$ and the directed path corresponding to the single edge $(v_{i,\alpha_{j}}, {y}_{i, \alpha_{j}})$ to $\mathcal{P}'$. Otherwise, $x_{i}$ appears as a negative literal in $C_{\alpha_j}$. Thus, none of the paths from $s_i$ to $t_i$ passes through the edge $(\overline{y}_{i, \alpha_{j}}, u_{i,\alpha_{j+1}})$. In this case, we add the directed path form $z_{0, \alpha_j}$ to $u_{i,\alpha_{j+1}}$ and the directed path corresponding to the single edge $(v_{i,\alpha_{j}}, \overline{y}_{i, \alpha_{j}})$ to $\mathcal{P}'$. Finally, assume $x_a$ and $x_b$ are the other two literals in $C_{\alpha_j}$. Observe that there are edges $e_a$ and $e_b$ connecting $z_{\alpha_j, 2}$ to the variable gadgets of $x_a$ and $x_b$, respectively. We add the two edge paths corresponding to edges $e_a$ and $e_b$ to $\mathcal{P}'$.
        
        Note that the edges of $G^*_f$ that are not covered by the paths introduced so far constitute paths with two edges, and there are exactly $\sum_{i=1}^{n} \ell_i - m = 2m$ such paths (after adding paths of interestingness score $W$, there are $\sum_{i=1}^{n} \ell_i$ many $2$-paths located in variable gadgets and $m$ of them are covered by paths starting at clauses). We add these 
        $2$-paths to the path partition in order to cover all edges in $G^*_f$. 
        
        In summary, $\mathcal{P}'$ contains
        \begin{itemize}
            \item $\sum_{i=1}^{n} \ell_i = 3m$ many $6$-paths,
            \item $m$ many $4$-paths,
            \item $2m$ many $2$-paths, and
            \item $3m$ many $1$-paths.
            
        \end{itemize}
        
        Therefore, for the interestingness score of $\mathcal{P}'$, we have 
        \begin{eqnarray*} 
        {\texttt{score}(\mathcal{P}')} & = & 
        m \cdot \left(\log{(5!)} + 3\log{2}\right) + 2m \cdot \log{3!} + 3m \cdot W  \\ 
        & = & m \cdot \left(\log{(5!)} + 5 \log{2} + 2 \log{3} + 3W\right) \\ 
        & = & \log t . 
        \end{eqnarray*} 
        
        For the other direction of the proof, assume we have a path partition $\mathcal{P}$ such that ${\texttt{score}(\mathcal{P})} \ge \log t$. First, we prove that for each variable $x_i$, there exist $\ell_i$ paths from $s_i$ to $t_i$ with interestingness score
        $W$ in $\mathcal{P}$. For the sake of contradiction, assume there is a type-$\mathrm{T}$ edge $(u_{i,\alpha_j}, t_i)$ in $G^*_f$ where $(u_{i,\alpha_j}, t_i)$ is not the $6$th edge on some path in $\mathcal{P}$. To upper bound the interestingness score of $\mathcal{P}$, for each edge $e$ in $G^*_f$, we consider the furthest possible position of the edge $e$ along all the paths going through $e$. In particular, by the construction of the graph $G^*_f$, the contribution of the edge $(u_{i,\alpha_j}, t_i)$ to the interestingness score of $\mathcal{P}$ is at most $29m \cdot \log{6}$.

        Since each path in $G^*_f$ has at most $6$ edges, $\cal P$ contains at most
        \begin{itemize}
            \item $\sum_{i=1}^{n} \ell_i - 1 = 3m - 1$ edges, each of which is at the $6$th position in some path of $\cal P$, 
            \item $\sum_{i=1}^{n} 2\ell_i + 1 = 6m + 1$ edges, each of which is at the $5$th position in some path of $\cal P$; exactly one of them has weight $29m$, 
            \item $\sum_{i=1}^{n} 2\ell_i = 6m$ edges, each of which is at the $4$th position in some path of $\cal P$,
            \item $m + \sum_{i=1}^{n} \ell_i = 4m$ edges, each of which is at the $3$rd position in some path of $\cal P$,
            \item $m + \sum_{i=1}^{n} \ell_i = 4m$ edges, each of which is at the $2$nd position in some path of $\cal P$, and
            \item $2m + \sum_{i=1}^{n} \ell_i = 5m$ edges, each of which is at the $1$st position in some path of $\cal P$.
        \end{itemize}
        
        Hence, recall that $W = \log{(6!)} + 29m\cdot\log{7}$,
        \begin{eqnarray*}
            {\texttt{score}(\mathcal{P})} & \le & 
            (3m-1) \cdot 29m \cdot \log{7} + 29m \cdot \log{6} + 6m \cdot \log{6} + 6m \cdot \log{5} \\ 
            & & + 4m \cdot \log{4} + 4m \cdot \log{3} + 5m \cdot \log{2} \\
            & = & m \cdot (\log{(5!)} + 5 \log{2} + 2 \log{3} + 3W) - 29m \cdot \log{\left(\frac{7}{6}\right)} \\ 
            & & + m \cdot (3\log{30} - \log{5} - 3\log{2} - 2\log{3}).
        \end{eqnarray*}
        Since $(3 \log{(30)} - \log{5} - 3 \log{2} - 2 \log{3}) - 29 \log{({7}/{6})} < 0$, we have 
        \[ {\texttt{score}(\mathcal{P})} < m \cdot (\log{(5!)} + 5 \log{2} + 2 \log{3} + 3W) = \log t , 
        \]
        which is a contradiction. Therefore, for each variable $x_i$, there exist $\ell_i$ paths from $s_i$ to $t_i$ with interestingness score $W$ in $\mathcal{P}$.

        Now, we claim that for each clause $C_{\alpha_j}$, there is a $4$-path that starts at the vertex $z_{\alpha_j, 0}$ and ends with a type-$\mathrm{U}$ edge. Since each type-$\mathrm{T}$ edge is in one of the $\ell_i$ paths from $s_i$ to $t_i$ in the corresponding variable gadget, the directed paths starting from $z_{\alpha_j,0}$ can have at most $4$ edges. Again, for the sake of contradiction, assume that there is a clause $C_{\alpha_j}$ such that the two edges in its gadget are not covered with a $4$-path ending at a type-$\mathrm{U}$ edge. For that particular clause, there are three different type-$\mathrm{U}$ edges that are not the $4$th edge on some directed path. We can now upper bound the value of $\texttt{score}(\mathcal{P})$, again by considering the furthest position of each edge $e$ in all the paths going through $e$.
        
        It follows that $\cal P$ contains at most
        \begin{itemize}
            \item $\sum_{i=1}^{n} \ell_i = 3m$ edges, each of which is at the $6$th position in some path of $\cal P$,
            \item $\sum_{i=1}^{n} \ell_i = 3m$ edges, each of which is at the $5$th position in some path of $\cal P$,
            \item $m + \sum_{i=1}^{n} \ell_i - 1 = 4m - 1$ edges, each of which is at the $4$th position in some path of $\cal P$,
            \item $m + \sum_{i=1}^{n} \ell_i  + 1 = 4m + 1$ edges, each of which is at the $3$rd position in some path of $\cal P$,
            \item $3m + \sum_{i=1}^{n} \ell_i = 6m$ edges, each of which is at the $2$nd position in some path of $\cal P$, and
            \item $6m + \sum_{i=1}^{n} \ell_i = 9m$ edges, each of which is at the $1$st position in some path of $\cal P$.
        \end{itemize}

        Therefore,
        \begin{eqnarray*}
            {\texttt{score}(\mathcal{P})} & \le & 
            3m \cdot 29m \cdot \log{7} + 3m \cdot \log{6} + (4m-1) \cdot \log{5} \\
            & & + (4m + 1) \cdot \log{4} + 6m \cdot \log{3} + 9m \cdot \log{2} \\
             & = & m \cdot (\log{(5!)} + 5 \log{2} + 2 \log{3} + 3W) - \log{5} + \log{4} \\
             & < & m \cdot (\log{(5!)} + 5 \log{2} + 2 \log{3} + 3W) \\
             & = & \log t. 
        \end{eqnarray*}
        This is a contradiction. Therefore, for each clause $C_{\alpha_j}$, there is a $4$-path in $\mathcal{P}$ that starts at the vertex $z_{\alpha_j, 0}$ and ends with a type-$\mathrm{U}$ edge.
        
        The truth-value assignment of the variables depends on the paths that $\mathcal{P}$ is made of. For each variable $x_i$, Since there are $\ell_i$ different paths with interestingness score $W$ in $G_i$, either none of the vertices labeled $y_{i,\alpha_j}$ or none of the vertices labeled $\overline{y}_{i,\alpha_j}$ are covered by these paths. In the former case, $x_i$ is assigned the truth-value $\mathrm{TRUE}$. In the latter case, the truth-value $\mathrm{FALSE}$ is assigned. Moreover, given a clause $C_{\alpha_j}$, assume the $4$-path covering the clause gadget of $C_{\alpha_j}$ ends with a type-$\mathrm{U}$ edge in the variable gadget of $x_i$. Then the truth-value assignment of $x_i$ makes $C_{\alpha_j}$ evaluate to TRUE.
        \renewcommand{\qedsymbol}{\scalebox{1.3}{$\triangle$}}
    \end{proof}

    Note that the formula $f$ has size $\Omega(m + n)$. The graph $G^*_f$ contains $O(n + m)$ vertices and edges, with each edge weight requiring $O(\log m)$ bits. The value of $t$ can be stored using $O(m^2)$ bits and computed in $O(m^4 \log m)$ time. Thus, constructing the instance is achievable in polynomial time. This concludes the proof of the reduction and shows that the IP problem is NP-hard, even when there are only two distinct edge weights.
\end{proof}


\section{Hardness of $k$-IP Problem} \label{sec:kip}

In this section, we prove the NP-hardness of the $k$-IP problem on directed acyclic graphs when $k \ge 3$. In fact, we show that the $k$-IP problem is NP-hard even when all the edge weights are equal to $1$. If the edge weights are equal to $1$, the weight of every $k$-path is equal to $\log{((k+1)!)}$. Consider the following decision version of the $k$-IP problem for a fixed integer $k \ge 3$.

\begin{tcolorbox}[colback=white, colframe=teal, sharp corners=all, boxrule=0.25mm, rounded corners=all, width=\textwidth]
\textbf{$k$-IP Problem (decision version)}:\\
\textbf{Input:} An edge-weighted directed acyclic graph $G=(V,E)$, a weight function $w : E \rightarrow \mathbb{N}$, and a positive integer $t$. \\
\textbf{Question:} Does there exist an edge-disjoint collection $\mathcal{P}$ of $k$-paths in $G$ such that the interestingness score of $\mathcal{P}$ is at least $\log{t}$?
\end{tcolorbox}

We will reduce, in polynomial time, the $(3,2)$-set-cover problem to the $k$-IP problem. The $(3,2)$-set-cover problem is NP-complete by a simple polynomial-time reduction from the vertex cover problem on cubic graphs which is known to be NP-complete \cite{DBLP:journals/tcs/GareyJS76}. The decision version of the $(3,2)$-set-cover problem is defined as follows.

\begin{tcolorbox}[colback=white, colframe=teal, sharp corners=all, boxrule=0.25mm, rounded corners=all, width=\textwidth]
\textbf{$(3,2)$-set-cover Problem}:\\
\textbf{Input:} A triple $(U, \mathcal{S}, \tau)$, where
\begin{itemize}
    \item $U$ is a set of $n$ elements,
    \item $\mathcal{S} = \{S_1, S_2, \dots, S_m\}$ is a set of $m$ subsets of $U$,
    \item $|S_i| = 3$ for each $1 \le i \le m$,
    \item each element of $U$ appears in exactly two sets in $\mathcal{S}$, and
    \item $\tau$ is an integer with $1 \leq \tau \leq m$.
\end{itemize}
\textbf{Question:} Are there at most $\tau$ sets in $\mathcal{S}$ such that their union is equal to $U$?
\end{tcolorbox}

Note that $m \geq n/3$. Let $U = \{x_1, x_2, \cdots x_n\}$ be a set of $n$ elements and let $\mathcal{S} = \{S_1, S_2, \cdots, S_m\}$ be a collection of $m$ subsets of $U$. For two indices $1 \le i \le n$ and $1 \le j \le m$, we say that a subset $S_j$ \emph{covers} $x_i$ if $x_i \in S_j$. We say that a subset $\mathcal{S}' \subseteq \mathcal{S}$ \emph{covers} $U$, if for each $x_i \in U$ ($1 \le i \le n$), there exist a set $S_j \in \mathcal{S}'$ such that $S_j$ covers $x_i$.

\begin{figure}
    \centering
    \includegraphics[width=0.96\linewidth]{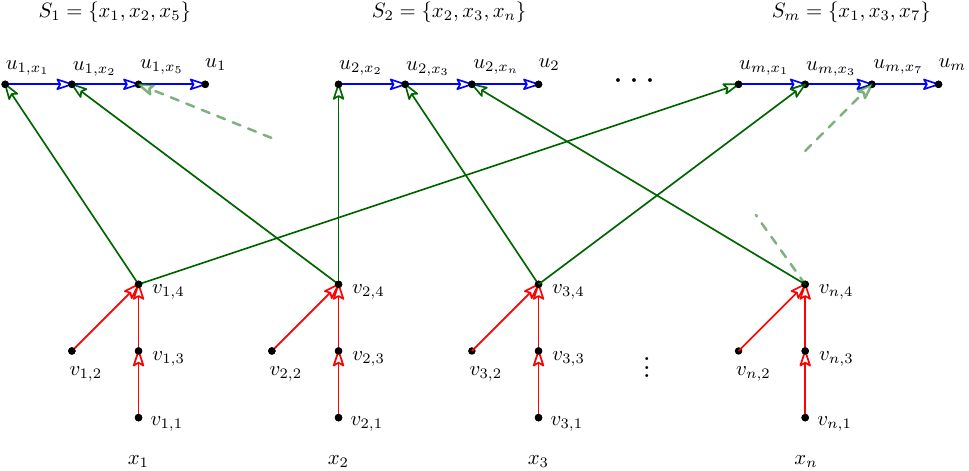}
    \caption{Construction of the graph $G$ for a given instance $(U, \mathcal{S},\tau)$ of the $(3,2)$-set-cover problem. The red, blue, and green edges in $G$ represent the elements of $U$, the elements of the sets in $\cal S$, and the membership of elements of $U$ in the sets of $\cal S$, respectively.}
    \label{fig:gadgets2}
\end{figure}

\begin{thm} \label{thm:kip}
    For any fixed integer $k \geq 3$, the $k$-IP problem is NP-hard.
\end{thm}

\begin{proof}
    First, we prove the statement for $k = 3$. Let $(U, \mathcal{S},\tau)$ be an instance of $(3,2)$-set-cover. 
    Using this instance, we construct a directed acyclic graph $G$ as follows; see \cref{fig:gadgets2} for an example.
    For each set $S_i = \{x_p, x_q, x_r\} \in \mathcal{S}$ with $1 \le p < q < r \le n$, consider four distinct vertices $u_{i,x_p},u_{i,x_q},u_{i,x_r}, u_{i}$ (type-$S$ vertices), and for each element $x_i \in U$ consider four distinct vertices $v_{i,1},v_{i,2},v_{i,3},v_{i,4}$. All edge weights in $G$ are equal to $1$ and the edge set of $G$ consists of the following directed edges:

    \begin{itemize}
        \item $(v_{i,1}, v_{i,3})$, $(v_{i,3}, v_{i,4})$, and $(v_{i,2}, v_{i,4})$, for each $1 \le i \le n$,  
        \item $(u_{i,x_p}, u_{i,x_q})$, $(u_{i,x_q},u_{i,x_r})$, and $(u_{i,x_r}, u_{i})$,  for each $1 \le i \le m$, where  $S_i = \{x_p, x_q, x_r\}$ with $1 \le p < q < r \le n$,

        \item $(v_{i,4}, u_{\alpha, x_i})$ and $(v_{i,4}, u_{\beta, x_i})$, for each $1 \le i \le n$ and $S_\alpha, S_\beta \in \mathcal{S}$ such that $x_i \in S_\alpha \cap S_\beta$.
    \end{itemize}
    Finally, we define the target parameter $t$ to be 
    \[ t = 4!^{2n + m - \tau} .     
    \]
    
    \begin{clm}
    The $(3,2)$-set-cover instance $(U, \mathcal{S}, \tau)$ has a cover of size at most $\tau$ if and only if $G$ has an edge-disjoint collection of $3$-paths with interestingness score at least
    \[ \log t = (2n+m - \tau ) \log (4!) . 
    \]
    \end{clm}

    \begin{proof}
        Since all edge weights in $G$ are equal to $1$, we only need to prove that the $(3,2)$-set-cover instance $(U, \mathcal{S},\tau)$ has a cover of size at most $\tau$ if and only if $G$ contains at least $2n + m - \tau$ edge-disjoint $3$-paths.
        
        Assume $(U,S,\tau)$ has a cover of size at most $\tau$, i.e., there is a set of indices $\mathcal{I} \subseteq \{1,2,\dots,m\}$ of size at most $\tau$ such that $\bigcup_{i\in\mathcal{I}} S_i = U$. We construct a set $\mathcal{P}$ of edge-disjoint 3-paths of size at least $2n+m-\tau$ for $G$. For each $i \in \{1,2,\dots,m\} \setminus \mathcal{I}$, if $S_i = \{x_p, x_q, x_r\}$ with $1 \le p < q < r \le n$, then we add the $3$-path $u_{i,x_p},u_{i,x_q},u_{i,x_r}, u_{i}$ to $\mathcal{P}$. For each element $x_i \in U$, we do the following. Let $\alpha \in \mathcal{I}$ and $\beta \in \{1,2,\dots,m\}$ such that $\alpha \neq \beta$, $x_i \in S_\alpha$, $x_i \in S_\beta$. We add the $3$-path that starts at $v_{i,2}$ and passes though the vertex $u_{\alpha, x_i}$ to $\mathcal{P}$. This path is unique by construction. We also add the $3$-path $v_{i, 1},v_{i,3},v_{i,4},u_{\beta,x_i}$ to $\mathcal{P}$. By construction, the paths in $\mathcal{P}$ are pairwise edge-disjoint and $|\mathcal{P}| \ge 2n + m - \tau$.

        For the other direction, let $\mathcal{P}$ be a maximum size collection of edge-disjoint $3$-paths of $G$. By the assumption, we have $|\mathcal{P}| \ge 2n+m-\tau$. Without loss of generality, we assume that no path in $\mathcal{P}$ starts at a vertex labeled $v_{i, 3}$ or $v_{i, 4}$ (such a path can be replaced with a path that starts at $v_{i,1}$ or $v_{i,2}$). Moreover, we can assume that for each $1 \le i \le n$, there are two paths in $\mathcal{P}$ starting at $v_{i,1}$ and $v_{i,2}$. Otherwise, either we can replace some path with a new path that starts at a vertex labeled $v_{i,1}$ or $v_{i,2}$, or we can add a new $3$-path starting at a vertex labeled $v_{i,1}$ or $v_{i,2}$ which contradicts the maximality of $\mathcal{P}$. Therefore, at least $m - \tau$ of the $3$-paths in $\mathcal{P}$ contains only vertices labeled $u$ and each such $3$-path corresponds to a set in $\mathcal{S}$. Let $\mathcal{J} \subseteq \{1,2,\dots,m\}$ be the indices of these sets in $\mathcal{S}$. Note that $|\mathcal{J}| \ge m - \tau$. Let $\mathcal{I} = \{1,2,\dots,m\} \setminus \mathcal{J}$. We claim that the subset of $\mathcal{S}$ with indices in $\mathcal{I}$ creates a $(3,2)$-cover for the set $U$. Recall that for each $1 \le i \le n$, there is a path in $\mathcal{P}$, that starts at the vertex $v_{i,2}$. The last two vertices of this $3$-path have label $u$. Let $S_j \in \mathcal{S}$ be the set corresponding to these two vertices. Note that, $j \in \mathcal{I}$. Therefore, each element $x_i$ is covered by a set whose index belongs to $\mathcal{I}$. Finally note that $|\mathcal{I}| \le m - (m - \tau) = \tau$.
        \renewcommand{\qedsymbol}{\scalebox{1.3}{$\triangle$}}
    \end{proof}
    
    The instance $(U, \mathcal{S},\tau)$ has size 
    \[ \Omega( m \log n + \log \tau ) = \Omega(m \log n) .    
    \]
    Since $m \geq n / 3$, the graph $G$ has size $O(m)$, and the number of bits in the binary representation of the value $t$ is 
    \[ (2n+m-\tau) \log (4!) = O(n+m) = O(m) . 
    \]
    Therefore, the reduction is polynomial in the input size.
    This concludes the proof of NP-hardness for the $k$-IP problem when $k = 3$. 

    For $k > 3$, the argument remains the same. However, in the construction of the graph $G$, we make the following modifications. For each $1 \leq i \leq n$, we append a path with $k-3$ edges to the vertices labeled $v_{i,1}$ and $v_{i,2}$. Similarly, for each $1 \leq j \leq m$ where $S_j = \{x_p, x_q, x_r\}$ with $1 \leq p < q < r \leq n$, we append a path with $k-3$ edges to the vertices labeled $u_{j,p}$.
    Moreover, we define the target parameter $t$ to be 
    \[ t = (k+1)!^{2n + m - \tau} . \]
    Since $k$ is fixed (i.e., $k$ is a constant), the reduction still takes polynomial time.
\end{proof}

\section{Conclusion}
\label{sec:future} 

We have shown that the IP problem and the $k$-IP problem (for any fixed $k \geq 3$) are NP-hard for edge-weighted directed acyclic graphs. 

We have not been able to prove that these problems are in NP. 
The key challenge lies in the fact that the interestingness score of a collection of paths involves sums of logarithms, which can be transcendental numbers. It is unclear if such sums can be compared to the target value $\log t$ in a polynomial number of bit-operations. One approach is to to replace $\texttt{score}$ by $2^{\texttt{score}}$, and replace $\log t$ by $t$. Even though there are no logarithms anymore, the number of bits in the binary representations becomes exponential. Proving whether these two problems are in NP is left as an open problem.

Given the known hardness results, it would be worthwhile to explore approximation algorithms for these problems. A straightforward greedy approach to the $k$-IP problem provides a $1/k$-approximation algorithm. 
Several heuristic algorithms have also been explored for these optimization problems, as discussed in the paper by Kalyanaraman et al.~\cite{DBLP:journals/jocg/KalyanaramanKK19} and in Vicuna's Master's Thesis~\cite{marcsthesis}. Designing better approximation algorithms for these problems remains an exciting direction for future research.

\newpage

\bibliographystyle{plainurl}
\bibliography{main}

\newpage

\end{document}